\documentclass[prl,aps,twocolumn,superscriptaddress,dvipsnames,amssymb]{revtex4}
\usepackage{graphicx}
\usepackage{amsmath,amssymb,amsthm,amsfonts,mathptmx}
\usepackage{color}
\usepackage[matrix,frame,arrow]{xypic}
\usepackage{times}
\input{Qcircuit}
\theoremstyle{plain}

\newtheorem{theorem}{Theorem}
\theoremstyle{definition}
\newtheorem{definition}{Definition}

\begin{document}
\title{Optimal Blind Quantum Computation}
\author{Atul Mantri}
\affiliation{Indian Institute of Science Education and Research (IISER) Mohali,
SAS Nagar, Sector 81, Mohali 140 306, Punjab, India} 
\affiliation{Centre for Quantum Technologies, National University of Singapore, 3 Science Drive 2, Singapore 117543}
\author{Carlos A. P\'erez-Delgado}\email{cqtcapd@nus.edu.sg}
\affiliation{Centre for Quantum Technologies, National University of Singapore, 3 Science Drive 2, Singapore 117543}
\author{Joseph F. Fitzsimons}\email{joe.fitzsimons@nus.edu.sg}
\affiliation{Centre for Quantum Technologies, National University of Singapore, 3 Science Drive 2, Singapore 117543}
\affiliation{Singapore University of Technology and Design, 20 Dover Drive, Singapore 138682}
\begin{abstract}
Blind quantum computation allows a client with limited quantum capabilities to interact with a remote quantum computer to perform an arbitrary quantum computation, while keeping the description of that computation hidden from the remote quantum computer. While a number of protocols have been proposed in recent years, little is currently understood about the resources necessary to accomplish the task. Here we present general techniques for upper and lower bounding the quantum communication necessary to perform blind quantum computation, and use these techniques to establish a concrete bounds for common choices of the client's quantum capabilities. Our results show that the UBQC protocol of Broadbent, Fitzsimons and Kashefi \cite{joe1}, comes within a factor of $\frac{8}{3}$ of optimal when the client is restricted to preparing single qubits. However, we describe a generalization of this protocol which requires exponentially less quantum communication when the client has a more sophisticated device.
\end{abstract}
\maketitle
The development of quantum computation promises the ability to solve computational problems which prove intractable for classical computers~\cite{shor1997polynomial, grover1996fast}. Quantum computers would also allow for the simulation of quantum systems that are not possible with present day technology~\cite{lloyd1996universal,kassal2008polynomial}. The study of quantum computation and information has also led to new insights on the fundamental quantum nature of physics \cite{peres2004quantum,lloyd02,braunstein2007quantum,zw10a,oreshkov2012quantum,fitzsimons2013quantum}. Recently, there has been growing interest in the nature of \emph{distributed} quantum computation \cite{buhrman,broadbent,danos}. Beyond the ability to shed new light onto the question of the nature of the (possible) advantage of quantum computation over classical \cite{ekert1998quantum,knill1998power,aaronson2011computational,bremner2011classical}, this area has important practical applications. Due to the difficulty in constructing large scale quantum computers, it is likely that the availability of such technology will be limited, at least at first. Hence, the ability to perform a quantum computation \emph{remotely} is of particular interest. More recently, the question has arisen whether it is possible to perform a quantum computation remotely in a \emph{blind} fashion. In a blind computation Alice gets Bob to perform a quantum computation for her without revealing the nature of the computation, or the input (up to some minimal \emph{leaked} information such as an upper bound on the size of the circuit/input) \cite{childs,arrighi,aharonov,joe1}. This notion mirrors the classical counterpart \cite{Feigen86}, though quantum mechanics appears to allow for encryption of a larger range of problems than is possible classically \cite{joe1}.

One of the least technologically demanding solutions to the problem of blind computation is the Universal Blind Quantum Computation (UBQC) protocol \cite{joe1,joe2}, which has recently been demonstrated experimentally in a quantum optics setting \cite{joe3}. This protocol has been shown to be both correct and secure, both in a stand-alone setting and as a cryptographic primitive \cite{dunjko2013composable}. However, the question arises of whether the protocol is \emph{optimal}. That is, is it possible to achieve correctness and security using \emph{less} resources. In this letter, we address this issue of how much quantum communication is necessary in order to achieve the blind evaluation of some secret unitary operation. To this end, we develop a framework to bound the resources of any possible blind computation. Fixing Alice's quantum capabilities, we define a figure of merit, $\Gamma(N)$, corresponding to the maximum number of quantum gates which can be hidden by any protocol which communicates $N$ qubits. We use a simple counting argument to bound $\Gamma(N)$ from above, and use a generalization of the blind computation protocol presented in \cite{joe1} to lower bound $\Gamma(N)$ by giving an achievable rate. We apply these techniques to obtain bounds in a number of physically realistic settings, including those for which blind quantum computing protocols have previously been proposed, as well as others motivated by the current state of quantum technology.

The paper is structured as follows. We begin by introducing our figure of merit, $\Gamma(N)$, and demonstrating a simple counting technique for bounding from above the rate at which gates can be hidden. We then proceed to introduce a generalization of the UBQC protocol, and consider its correctness and blindness. We use these techniques to examine various limitations on Alice's computational abilities. We use the generalized blind computation protocol and the parameter counting argument to bound $\Gamma(N)$ from below and above, respectively, in each setting. We conclude with a discussion of the universality of the generalized protocol in the settings under consideration.

As the exact relationship between BQP and NP remains unknown, there is in fact no proof that the decision problems answerable with a quantum computer cannot be hidden from a remote server using purely classical means, as is achieved for certain problems in \cite{Feigen86}. However, quantum computation does more than answer decision problems. It manipulates quantum states in a continuous way, and such a computation cannot be completely hidden using purely classical communication \endnote{To see this, consider a protocol which consisted entirely of classical communication. Given enough time, Bob could rerun the protocol many times using the transcript of a single run with Alice, and perform tomography on the output state. This would limit the possible computations which could have been performed by Alice to a discrete set. Hence continuous rotations cannot be completely hidden, and the computation is revealed up to Alice's decoding operations.}. 

\begin{definition} We define a \textit{single parameter gate} to be any gate parameterised by a single real variable which, for varying values of the parameter, maps an input state which is in a fixed maximally entangled state with an ancilla register onto states which collectively lie on a one dimensional curve of finite length in the Hilbert space of the system. Then, for a particular choice of Alice's apparatus, we define $\Gamma(N)$ to be the maximum number of such single parameter gates which can by encoded across $N$ transmitted qubits given the limitations of Alice's device.
\end{definition}
A simple example of such a single parameter gate is a Pauli rotation through an arbitrary angle. Given the above definition, $\Gamma(N)$ can be bounded using a well known result from topology.
\begin{theorem} For a fixed choice of Alice's apparatus, if any $N$ qubit output state which Alice can produce lies of a manifold of real dimension $D$, then $\Gamma(N)\leq D$. \label{thm:upper-bound}
\end{theorem}
\begin{proof}
For a quantum computation composed of $\Gamma(N)$ single parameter gates, provided that no gates are redundant, the possible output states correspond to points on a manifold of real dimension $\Gamma(N)$, since each such gate increases the real dimension of the manifold by at most one. It is well established in topology that a manifold of finite dimension cannot be continuously mapped into a manifold of lower dimension \cite{topology}. However, the input states received by Bob lie on a manifold of dimension $D$. Hence, since any operation Bob can perform is necessarily continuous due to the linearity of quantum mechanics, $\Gamma(N)\leq D$.
\end{proof}
This theorem implies that for a fixed choice of her quantum capabilities, by bounding $D$ by counting the independent continuous parameters necessary to describe states produced by Alice, it is possible to place an upper bound on $\Gamma(N)$. 

We now turn our attention to establishing a lower bound on $\Gamma(N)$, by presenting a generalization of the Broadbent {\it et al.} protocol. We assume that Alice has the ability to generate input states randomly chosen from some set, $\Phi$, which she can then send (perfectly) to Bob, and that Bob has access to a full quantum computer. 

We will consider only sets $\Phi$ which can be generated in the following way. Take a set of diagonal unitary operators $\mathcal{D}=\{D_k\}_k$ which forms a group under multiplication, and which has the additional property that $\left( X^{c^{1}} \otimes \dots \otimes X^{c^{n}} \right) \mathcal{D} \left( X^{c^{1}} \otimes \dots \otimes X^{c^{n} } \right) = \mathcal{D}$ for all $\{c^j\} \in \{0,1\}^{\otimes n}$. We then define
\[
\Phi = \left\{\ket{\phi_k} : \left( Z^{r^{1}_{k}} \otimes \dots \otimes Z^{r^{n}_{k}} \right) D_k |+\rangle^{\otimes n}, D_k \in \mathcal{D}, r_k^j \in \{0,1\}\right\}.
\]
The generalized UBQC (GUBQC) protocol is then as follows:
\begin{enumerate}
\item Alice chooses an ordered set $\{U_i\}$ of $m$ operators from $\mathcal{D}$ such that $H^{\otimes n} U_{m} H^{\otimes n} U_{m-1} ... H^{\otimes n} U_1 \ket{+}^{\otimes n}$ corresponds to her desired computation. 
\item For every $1\leq i \leq m$ Alice chooses $D_i \in_R \mathcal{D}$ (uniformly at random) and $\{r_i^k : 1 \leq k \leq n\} \in \{0,1\}^{\otimes n}$. She then prepares the $n$-qubit state $|\phi_{i} \rangle = \left( Z^{r^{1}_{i}} \otimes \dots \otimes Z^{r^{n}_{i}} \right) D_{i} |+\rangle^{\otimes n}$, and sends it to Bob. Bob stores this state in the $i^\text{th}$ register of his quantum computer.
\item For all $1 \leq i < m-1$ and $1\leq j\leq n$, Bob interacts the $j^\text{th}$ qubit of register $i$ to the $j^\text{th}$ qubit of register $i+1$, by applying a controlled-$Z$ gate between them. 
\item For each step $1\leq i \leq m$:
\begin{enumerate}
\item Alice sends Bob a classical description of the operator $ C_{i} = D_{i}^{\dag} U_{i}' ,$ where 
\[~~~~~~~~~~~~~~~~~U_{i}' = \left( X^{c^{1}_{i}} \otimes \dots \otimes X^{c^{n}_{i}} \right) U_{i} \left( X^{c^{1}_{i}} \otimes \dots \otimes X^{c^{n}_{i} } \right).\]
Here $c^{j}_{i} = \oplus_{k \in \mathcal{S}_{i-1}} s^{j}_{k}$ and $s^{j}_{k} = b^{j}_{k} \oplus r^{j}_{k}$, $b^{j}_{i}$ is the measurement result on the $j^{th}$ qubit of $i^{th}$ $n$-qubit state. Here $\mathcal{S}_i = \lbrace i, i-2, \ldots \rbrace$ captures the measurement dependency structure between registers. We use the convention that $c^{j}_{1} = 0 , \forall j$ and $s^{j}_{i} = 0~\forall j$ for $i\leq 1$.
\item Bob applies $C_{i}$ on the $i^\text{th}$ register of his system. He then performs a measurement of each qubit in that register in the Hadamard basis (i.e., $\lbrace |+\rangle , |-\rangle \rbrace$). He sends the measurement result $b_{i} = \lbrace b_{i}^{j} \rbrace$ to Alice, with the convention that $|+\rangle$ corresponds to the outcome 0, while $|-\rangle$ corresponds to 1.
\item Alice sets the value of $s^{j}_{i} = b^{j}_{i-1} \oplus r^{j}_{i-1}.$
\end{enumerate}
\item In the case of classical output, Alice takes the ordered set $c_{m+1} = \{c_{m+1}^j: 1\leq j \leq n\}$ as the output of the computation. In the case of quantum output, Bob returns the final register to Alice, without measurement, after applying $C_m$ followed by Hadamard gates on each qubit. Alice then performs the last set of correction operators herself, by applying $\bigotimes_j Z^{c_m^j} X^{c_{m-1}^j + r_m^j}$ to the state she receives.
\end{enumerate}
If both parties follow the protocol then the result corresponds to Alice's desired computation, as shown below. 
\begin{theorem} If Alice and Bob follow the steps as set out in the protocol, then the output received by Alice is  $H^{\otimes n} U_{m} H^{\otimes n} U_{m-1} ... H^{\otimes n} U_1 \ket{+}^{\otimes n}$ in the case of a quantum output, or the result of measuring this state in the computational basis otherwise.
\end{theorem}
\begin{proof} As $D_i$ is a diagonal operator, it commutes with the controlled-$Z$ operators used to entangle Bob's registers. Hence, the net effect of the protocol is identical to the case where $D_i = U_i'$ and $C_i$ is the identity, independent of the specific choices actually made by Alice for $D_i$. Therefore the action of $C_{1}$ in first layer of protocol can be equivalently written as $\left( Z^{r^{1}_{1}} \otimes \dots \otimes Z^{r^{n}_{1}} \right) U_{1} |+\rangle^{\otimes n}$.
By using the one bit teleportation circuit \cite{gottesman}, as shown in Fig \ref{fig:single}, the effect of measurement as if the initial state sent for the second register was prepared in the state $\bigotimes_{j=1}^n \left(X^{b_1^j + r_i^j} Z^{r_{2}^j}\right) U_2 H^{\otimes n} U_1 \ket{+}^{\otimes n}$,
and no previous layers existed.
\begin{figure}
\[
\Qcircuit @C=.5em @R=.2em {
& \lstick{|\phi\rangle} & \ctrl{4} & \qw & \gate{H} & \qw & \meter & \rstick{m} \\
\\
\\
\\
& \lstick{|+\rangle} & \control \qw & \qw & \qw & \qw & \qw & \qw & \qw & \qw & \rstick{X^{m}H|\phi\rangle}
}
\]
\caption{The single qubit teleportation protocol.\label{fig:single}}
\end{figure}
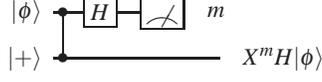

Applying this equivalence recursively we obtain the effective initial state of the final register as $\bigotimes_{j=1}^n \left(X^{c_m^j} Z^{c_{m-1}^j + r_{m}^j}\right) U_m H^{\otimes n} U_{m-1} ... H^{\otimes n} U_1 \ket{+}^{\otimes n}$.
Thus, in the case of classical output, since the $X$ operators commute with the measurement and $c_{m+1}$ incorporates corrections for the $Z$ byproducts, the result of the computation is as desired. Similarly, for quantum output, after the Hadamard gates are applied by Bob, Alice's correction exactly cancels the $X$ and $Z$ by products to result in final state of Alice's register of $H^{\otimes n} U_m H^{\otimes n} U_{m-1} ... H^{\otimes n} U_1 \ket{+}^{\otimes n}$.
\end{proof}
Having shown that the output of the protocol is indeed as expected in the case that Bob follows the protocol, we now turn our attention to the issue of blindness. In proving that our protocol is indeed blind, we base the definition of blindness on that used in \cite{joe1} and \cite{joe2}.
\begin{theorem} The GUBQC protocol is blind while leaking at most $n$ and $m$. 
\end{theorem} 
\begin{proof} In order for the GUBQC protocol to be blind while leaking at most it is necessary that two conditions hold: 1) the distribution of classical information obtained by Bob during a run of the protocol must depend only on $n$ and $m$, and 2) given the distribution of above classical information, as well as $n$ and $m$, the state of the quantum system obtained by Bob from Alice is fixed.

The information received by Bob in the protocol are the circuit dimensions $n$ and $m$, $m$ different $n$-qubit quantum states $\ket{\phi_i}$, and classical descriptions of each $C_i$. We first note that for a given $U_i'$ the distribution of $C_i$ is uniformly random over elements of $\mathcal{D}$, since $D_i^\dagger$ is randomly chosen, and $C_i = D_i^\dagger U_i'$. Thus the first criterion is satisfied. Although the quantum states $|\phi_{i} \rangle = \left( Z^{r^{1}_{i}} \otimes \dots \otimes Z^{r^{n}_{i}} \right) D_{i} |+\rangle^{\otimes n}$ appear correlated with $C_i$, this is in fact not the case. As all operators in $\mathcal{D}$ are diagonal, they commute with Pauli $Z$ operators, and so $|\phi_{i} \rangle = D_{i} \left( Z^{r^{1}_{i}} \otimes \dots \otimes Z^{r^{n}_{i}} \right) |+\rangle^{\otimes n}$. As $r_i$ is chosen uniformly at random, we must average over this secret parameter to determine the reduced density matrix for the quantum state Bob receives, which results in the maximally mixed state for any fixed $C_i$ and $U_i'$. Thus the second criterion is also satisfied.
\end{proof}

We now have the tools in place to bound $\Gamma(N)$ for a specific choice of Alice's quantum capabilities. Theorem \ref{thm:upper-bound} allows us to upper bound $\Gamma(N)$ by determining the dimensionality of any manifold containing all possible states generated by Alice's device. The GUBQC protocol, on the other hand, represents a concrete blind computation protocol, and a lower bound on $\Gamma(N)$ can be obtained for a specific setting by identifying a suitable set of states $\mathcal{D}$. We now calculate the bounds for four specific settings corresponding to various limitations on Alice's quantum capabilities.

\textit{1. Restriction to preparing separable qubit states:} The first case we consider is where Alice is restricted to transmitting individual qubits prepared in a separable state. This setting places very little technological requirements on Alice, as she need only be able to prepare and send a single qubit at a time. Similar capabilities have already been widely demonstrated in the context of quantum key distribution \cite{schmitt2007experimental}.
The upper bound on $\Gamma(N)$ is straightforward to calculate in this instance, since all single qubit pure states reside on a two-dimensional surface (the surface of the Bloch sphere). Thus separable states of $N$ qubits lie on a surface of $2N$ dimensions, and hence by Theorem \ref{thm:upper-bound} we have $\Gamma(N) \leq 2N$.

As this corresponds to the setting considered by Broadbent \textit{et al} \cite{joe1}, we can use the UBQC protocol presented in that paper to place a lower bound on $\Gamma(N)$. The explicit gate construction they present encodes up to 3 single parameter gates for every 4 qubits sent from Alice to Bob. However, a general measurement pattern on the same graph obtains a single parameter gate (in the form of a $Z$ rotation followed by a Hadamard gate) for every qubit sent from Alice to Bob, thus lower bounding $\Gamma(N)$ by $N$. Thus the UBQC protocol is within a factor of $\frac{8}{3}$ of optimality, and we have $N \leq \Gamma(N)\leq 2N$.

\textit{2. Restriction to preparing separable $k$-qubit states} We now consider a generalization of the previous setting, where we instead allow Alice to prepare entangled states of $k$ qubits at a time, which are then send to Bob. This corresponds to the situation where Alice can prepare entangled states of a certain size, but cannot store and interact the qubits she produces. Physically, this is motivated by quantum optics, where production and transmission of entangled states can be achievable (for example by parametric down conversion \cite{kwiat1995new}) with significantly less effort than is required to interact photons.

In general the quantum state of $k$ qubits can be written as $|\psi_{N}\rangle = \alpha_{1}|00....0\rangle + \alpha_{2}|00....1\rangle + ..... + \alpha_{2^{k}}|11....1\rangle$, where $\alpha_i$ are complex numbers. Since these coefficients are normalized such that $\sum_i |\alpha_i|^2 = 1$, and global phases can be neglected, such states lie on a surface of dimensionality $2^{k+1} - 2$. Thus, by Theorem \ref{thm:upper-bound}, we have $\Gamma(N) \leq \frac{2N}{k}(2^k - 1)$.
Note that since Alice can prepare any $k$ qubit state, she can necessarily prepare states of the form $D_i \ket{+}^{\otimes k}$, as needed for the GUBQC protocol, where $\mathcal{D}$ (the set from which all $D_i$ are drawn) is taken simply to be the set of tensor products of $\frac{n}{k}$ arbitrary diagonal unitary operators on $k$ qubits, where for simplicity we will take $n$ to be an integer multiple of $k$. Thus each $U_i$ contains $\frac{n}{k}(2^k - 1)$ single parameter gates. As there are $\frac{N}{n}$ such $U_i$ performed, $\Gamma(N)$ is lower bounded by $\frac{N}{k}(2^k - 1)$. In the case where $k=N$ this reduces to $\left(2^{N}-1\right) \leq \Gamma(N) \leq 2\left(2^{N}-1\right)$ meaning that an exponential number of single parameter gates can be hidden. 

\textit{3. Restriction to commuting unitary operators:} We now consider the case where Alice is restricted to applying operators from a commuting set to a fixed input, which we will assume to be the Hadamard transform of one of the common eigenstates of this set of operators, as in \textit{instantaneous quantum computation} \cite{shepherd2008instantaneous}. By using exactly the same choice for $\mathcal{D}$ as in the previous case we obtain a similar lower bound, i.e. $\left(2^{N}-1\right)$. Here, unlike in previous settings, our parameter counting argument yields a matching upper bound, since the set of states producible by Alice's apparatus lie on a manifold of exactly $\left(2^{N}-1\right)$ dimensions. 

However we can generalize this case by assuming that Alice can apply no more than $f(N)$ commuting single parameter gates in a given run of the protocol. This restriction is motivated by the desire to consider settings where Alice is required to perform only computationally efficient operations. Trivially, Theorem \ref{thm:upper-bound} implies that $\Gamma(N) \leq f(N)$. Using the same protocol as for the previous case, as long as $n \geq \log_2 f(N)$ and $\frac{f(N)}{m-1}$ is an integer, it is possible to choose $\mathcal{D}$ so that it encodes exactly $\frac{f(N)}{m-1}$ single parameter gates in each $U_i$, and hence hides $f(N)$ gates. As we then have matching upper and lower bounds, this implies that $\Gamma(N) = f(N)$.

\textit{4. Restriction on quantum memory:} The last case we will consider is where Alice possesses a quantum computer with a finite memory, as in \cite{aharonov}, and can send qubits individually while keeping others in memory and replacing the transmitted qubit. The unitary operators that can be used in this case are restricted to be $k$-sized. In this case, the most general operation Alice can perform is to iteratively perform a unitary across her entire register and then transmit a single qubit to Bob, replacing the sent qubit. When she reaches the last $k$ qubits she can transmit them all at once, since any unitary applied to these qubits between transmissions could have been absorbed into a previous operation.

In this case, it is not viable to directly calculate the exact dimensionality of the lowest dimensional manifold on which all of the states producible by Alice's device lie. Instead we upper bound this quantity by simply counting the number of independent single parameter gates Alice can perform before she has transmitted the last qubit to Bob. A general $k$-qubit unitary operator can be decomposed into exactly $2^{2k} -1$ free parameters. Since Alice applies such an every time she replaces a transmitted qubit the resulting states lie on a manifold of dimension $(N-k+1)(2^{2k} -1)$. However this bound can be improved by noting that a unitary operation performed on the set of $k-1$ common qubits between rounds could be absorbed into the $k$-qubit unitary in either the preceding or subsequent rounds. Eliminating this redundancy reduces the number of parameters by $2^{2(k-1)} - 1$ for a total of $N-k$ rounds. This leads to an improved bound of $\Gamma(N) \leq (N-k)(2^{2k} - 2^{2(k-1)}) + 2^{2k} -1$.

Turning to the lower bound, we consider what would happen if Alice initially prepared qubits in the state $\ket{+}$ and then applied only diagonal operations. In this case it is possible to exactly count the number of independent single parameter gates applied to the initial state. The unitary applied to the initial $k$ qubits has $2^k - 1$ free parameters, while each subsequent unitary must act non-trivial on the replaced qubit in order to be distinct from previous operations, and hence has $2^{k-1}$ free parameters. Thus states produced in this way lie on a manifold of dimension $(2^k - 1) + 2^{k-1}(n-k)$, and so $\frac{N}{n}\left(2^{k-1}(n-k+2) - 1\right) \leq \Gamma(N) \leq (N-k)(2^{2k} - 2^{2(k-1)}) + 2^{2k} -1$.

The four settings considered above are intended to cover the most obvious choices of Alice's apparatus, however we note that the technique used to upper bound $\Gamma(N)$ can readily be applied to any device. The GUBQC protocol is not quite as general, as it requires Alice to produce states from a set with a certain mathematical structure. Nonetheless, we expect that the GUBQC protocol can be adapted to most settings of practical interest, through a suitable choice of $\mathcal{D}$. Although we have not addressed the question of universality for the GUBQC protocol, in the three settings where it is used here, as long as Alice's system is of at least two qubits, in all cases the identity and CZ gates, as well as arbitrary local $Z$ rotations lie in $\mathcal{D}$. As each $U_i$ can be chosen from this set of gates arbitrarily, when the fixed Hadamard gates are taken into account, the set of operations is universal for quantum computation.

Although we have found that existing protocols are close to optimal for the first setting, in the other three settings the GUBQC protocol can hide significantly more quantum gates per qubit communicated than prior protocols, and in some cases requires exponentially less quantum communication. Further, in these cases (cases 2 and 3), the GUBQC protocol is within a factor of two of being optimal.

After the initial preparation of this manuscript, the authors became aware of a recent proposal from Giovanetti, Maccone, Morimae and Rudolph \cite{GMMR} for blind computation in the first setting which claims optimality. We note that their protocol is only optimal when quantum and classical communication are treated equally, and is not optimal from the point of view of quantum communication alone. Indeed, for their scheme, we obtain a value of $\Gamma(N) \propto N / \log n$, compared to $\Gamma(N) = \frac{3}{4}N$ for the UBQC scheme of \cite{joe1}. Nonetheless, we believe their protocol represents an interesting new approach to blind computation.

\textit{Acknowledgements --} JF and CPD acknowledge support from the National Research Foundation and Ministry of Education, Singapore. %This material is based on research funded in part by the Singapore National Research Foundation under NRF Award NRF-NRFF2013-01. 
\bibliographystyle{apsrev} 
\bibliography{bqc} 
\end{document}